\documentclass[twocolumn]{autart}    

\usepackage{amsmath,amssymb,epsfig,graphicx,lscape,epstopdf}
\usepackage{cite}
\usepackage{color}
\usepackage{xcolor}
\usepackage{amsmath,amssymb,amsfonts}
\usepackage{algorithm}
\usepackage{algpseudocode}
\usepackage{graphicx}
\usepackage{textcomp}
\usepackage{caption}
\usepackage{amsmath}
\usepackage[justification=centering]{caption}
\usepackage{subfig}
\usepackage{float}
\usepackage{multirow}

{
	
\newtheorem{theorem}{\textbf{Theorem}}[section]
\newtheorem{lemma}[theorem]{\textbf{Lemma}}
\newtheorem{corollary}[theorem]{\textbf{Corollary}}

\newtheorem{proof}{Proof}

\newtheorem{property}{\textbf{Property}}

\newtheorem{remark}{\textbf{Remark}}

}

\usepackage[justification=centering]{caption}
\allowdisplaybreaks[4]
\usepackage{graphicx}                                      

\begin{document}

\begin{frontmatter}

\title{Novel Nussbaum-Type Function based Safe Adaptive Distributed Consensus Control with Arbitrary Unknown Control Directions}

\thanks[footnoteinfo]{Corresponding author.}

\author[label0]{Dan Qiao
},\author[label0]{Zhaoxia Peng\thanksref{footnoteinfo}}\ead{pengzhaoxia@buaa.edu.cn},
\author[label2]{Guoguang Wen},
\author[label11]{Tingwen Huang}

\address[label0]{School of Transportation Science and Engineering, Beihang University, Beijing, 100191,P.R.China}
\address[label2]{Department of Mathematics, Beijing Jiaotong University, Beijing 100044, P.R.China}
\address[label11]{Science Program, Texas A\&M University at Qatar, Doha, Qatar}


\begin{abstract}                          
Existing Nussbaum function based methods on the consensus of multi-agent systems require (partial) identical unknown control directions of all agents and cause dangerous dramatic control shocks. This paper develops a novel saturated Nussbaum function to relax such limitations and proposes a Nussbaum function based control scheme for the consensus problem of multi-agent systems with arbitrary non-identical unknown control directions and safe control progress. First, a novel type of the Nussbaum function with different frequencies is proposed in the form of saturated time-elongation functions, which provides a more smooth and safer transient performance of the control progress. Furthermore, the novel Nussbaum function is employed to design distributed adaptive control algorithms for linearly parameterized multi-agent systems to achieve average consensus cooperatively without dramatic control shocks. Then, under the undirected connected communication topology, all the signals of the closed-loop systems are proved to be bounded and asymptotically convergent. Finally, two comparative numerical simulation examples are carried out to verify the effectiveness and the superiority of the proposed approach with smaller control shock amplitudes than traditional Nussbaum methods.
\end{abstract}

\end{frontmatter}

\section{Introduction}

Since the pioneering works \cite{vicsek,olfati,ren2005}, distributed multi-agent systems (MASs) have witnessed a flourishing development due to their robustness, flexibility, and collaboration ability. In the consensus problem of MASs, all agents asymptotically reach an agreement on their states by relying only on local interaction protocols with their neighborhoods. Subsequently, consensus problems of MASs have been discussed in a lot of real applications, including smart grids\cite{smartgrid}, multi-robots formation control \cite{fuMRS2014}, and social networks \cite{ye2020tac} to name a few. However, most of these results are developed based on the assumption that the control directions of the inputs, namely, the motion directions, in the systems are known, which may not be obtained directly as a \textit{priori} information in the practical applications such as the autopilot design of time-varying ships \cite{autoship} and the uncalibrated visual servo \cite{calibri}.

Considering the fail or destruction led by ignoring the unknown control directions (UCDs), this problem with both theoretical and practical significance has been extensively discussed in numerous systems over the years. In \cite{nussbaum1983}, a nonlinear Nussbaum-type control gain in the form of $f(k)\sin{k}$ was originally proposed to solve the unknown signs of the control inputs for MIMO systems. Even if other methods such as nonlinear PI control  \cite{haris2016tacPI,haris2018tacPI,haris2019PIerror} and switching mechanism \cite{huang2020IJRNCuncontinous} have been applied to deal with the UCDs problem recently, the Nussbaum-type function based control scheme is still the most effective method and has been widely applied in multiple SISO systems and MIMO systems such as \cite{1998single}.

However, the traditional Nussbaum functions cannot be directly applied to the stability analysis of MASs since the interactions among agents with multiple arbitrary UCDs will bring many new problems to the systems. In \cite{renwei2014tac}, Chen \textit{et al.} firstly discussed the UCDs problem in MASs, and proposed a significant novel Nussbaum function $\cosh{(\lambda k)}\sin{k}$ instead of the previous standard Nussbaum functions $k^2\cos{k}, {e}^{k^2}\cos{k}$, and $k^2\sin{k}$ to overcome the obstacles of the coexistence of the multiple Nussbaum functions in the same conditional inequality. 

Since then, based on the Nussbaum functions and the proof framework established in \cite{renwei2014tac}, a large number of results addressing UCDs with various other issues have emerged continuously in MAS. In \cite{Pengjm2014SCL,Su2015tac,Ding2015auto,liulu2015tac,liulu2016tac,Ma2017math,Huangjie2018tac}, the adaptive stabilization problem, asymptotic regulation problem, and output regulation problem with UCDs were investigated in numerous multi-agent systems including first-order agents, second-order agents, high order agents, etc. In \cite{HTWMahui2019ppc,LiTieShanWangWei2020PPC}, the convergent rate and the tracking errors of the cooperative control algorithms for high-order systems with unknown nonlinearities were controllable by leveraging prescribed performance methods and barrier Lyapunov functions. In \cite{LiHY2019trigger,Mahui2019TFuzzytrigger,guanZH2020trigger}, the event-triggered mechanism was employed for MAS with UCDs to realize the consensus objective with reduced signal transmission and resource-saving communication methods. In \cite{LiHYlurenquan2021TNNLSFTC}, Dong \textit{et al.} investigated the finite-time fault-tolerant consensus control for the nonlinear MAS with UCDs against the actuator bias fault and output deadzone. It is noted that the discussion of the Nussbaum functions itself was not enough in the above works. There are still two open problems for Nussbaum functions that need to be investigated, i.e., the control shock and the identical UCDs assumption.

On the one hand, the rapid growth of the exponential term in Nussbaum functions and the sign-changing of the trigonometric term will cause a sharp change in the control input, a.k.a., the control shock which greatly affects the transient performance of the system and causes severe shocks in the initial stage of control as shown in  \cite{renwei2014tac}, Fig. 2(a) and Fig. 3(b),  \cite{HUANG2018auto} Fig. 7,  \cite{LiHY2019trigger}, Fig. 4, etc. For such a problem, Chen \textit{et al.} transferred the traditional amplitude-elongation Nussbaum functions into a novel time-elongation one which could eliminate the control shock by designing a saturated piecewise function to suppress the upper bounds of the gains with identical UCDs \cite{chenci2016Saturated}.

On the other hand, most of the results above were based on the strict assumption that all the unknown control directions in the systems should be identical, which is a huge limitation in practical applications. Consequently, Chen \textit{ et al.} relaxed this requirement to partially non-identical UCDs circumstance by designing two piecewise Nussbaum functions to approximate the unknown gains of two groups of agents with different signs \cite{chenci2017tac}. This breakthrough which needs to know the sign of one group agents' UCDs has been widely applied in  time-varying MAS and nonaffine MAS with actuator nonlinearities and output constraints recently \cite{Xiekan2018cyber,BoFan2019TACpartial}.

Furthermore, in \cite{HUANG2018auto}, Huang \textit{et al.} completely solved  the hypothesis of the identical UCDs by proposing a set of Nussbaum functions in the same form with different frequencies for each agent. The exponential growth rate of ${e}^{\alpha k}$ is used to ensure that there must be a minimum value in the same interval of each agent to construct contradictions, where the $sin$ terms with different lengths are all negative. Similarly, in \cite{Zhang2021TAC}, a novel Nussbaum function has been developed with different frequencies to deal with a class of interconnected uncertain systems that are composed of multiple subsystems with different yet unknown control directions. The input shock cannot be avoided in these approaches. To the best of our knowledge, it remains still an open problem to develop a novel Nussbaum function to tackle the rigid limitation of identical UCDs and transient performance under a unified function.

In this paper, for the first time, we address this problem by proposing a novel Nussbaum function to allow arbitrary non-identical UCDs and to avoid control shocks. It is proved that the proposed Nussbaum function can guarantee the asymptotical consensus of multi-agent systems with a smoother and safer control progress than previous works.  The main contributions of this paper are the following. (i) A novel time-elongation saturated Nussbaum function with limited amplitude bound of Nussbaum gain and different frequency for each agent such that the multiple Nussbaum functions with interconnections could be analyzed in a single inequality.  (ii) The dramatic control shock  and the assumption of (partial) identical UCDs have been simultaneously relaxed by establishing a novel analysis framework for the inequality analysis of multiple Nussbaum functions, which leverages the proportional relationship between different Nussbaum functions. This idea could help inspire the future development of a family of Nussbaum functions for multi-agent systems. (iii) Asymptotic average consensus of the multi-agent system is achieved in the presence of uncertain control gains, which has wild applications in real engineering areas.

The organization of this paper is as follows. Preliminaries and problem formulation are provided in Section \ref{sec2}. The proposed novel Nussbaum functions are presented in Section \ref{sec3}. Two distributed adaptive Nussbaum-based consensus strategy is proposed in Section \ref{sec4}. Numerical simulations and comparative study are given in Section \ref{sec5}, and Section \ref{sec6} concludes our work in this paper.

\section{Preliminaries and Problem Formulation}\label{sec2}

\subsection{Graph Theory}
Consider a leaderless multi-agent system consisting of multiple agents. Index the agents as $1, 2, \dots, N$. The interaction relationship among the agents is represented by an undirected graph $\mathcal{G}=[\mathcal{V},\mathcal{E},\mathcal{A}]$, in which $\mathcal{V}=\left\{ {{n}_{1}},\ldots ,{{n}_{N}} \right\}$ represents a set of nodes, $\mathcal{E}\subseteq {{\mathcal{V}}_{N}}\times {{\mathcal{V}}_{N}}$ represents a set of edges, and $\mathcal{A}=\left[ {{a}_{ij}} \right]\in {{\mathbb{R}}^{N\times N}}$ represents the adjacency matrix of graph $\mathcal{G}$. Here, ${{a}_{ij}}$ is the weight of the link $\left( {{n}_{i}},{{n}_{j}} \right)$ with ${{a}_{ij}}>0$ if and only if  $\left( {{n}_{i}},{{n}_{j}} \right)\in\mathcal{E}$, ${a}_{ij}=0$  otherwise. Note that ${a}_{ij}={a}_{ji}, \forall {i}\neq{j}$  and ${a}_{ii}=0$  since the graph $\mathcal{G}$  is undirected with no self-edges. The Laplacian matrix ${L}=\left[{l}_{ij}\right]\in {{\mathbb{R}}^{N\times N}}$  is defined with ${l}_{ij}=-{a}_{ij}, \forall {i} \neq {j}$  and ${{l}_{ii}}=-\sum\limits_{i=1}^{N}{{{a}_{ij}}},\quad i,j\in 1,2,\ldots,N $. For graph $\mathcal{G}$, the following lemma is satisfied.
\begin{lemma}
	Assume that the topology graph $\mathcal{G}$ is undirected. Then, there must be at least one zero eigenvalue and all other nonzero eigenvalues are positive in $L$. Besides, the matrix $L$  is symmetric positive semi-definite for the connected undirected graph $\mathcal{G}$.
\end{lemma}

\subsection{Nussbaum Function}
Since being proposed in \cite{nussbaum1983}, Nussbaum functions have been widely applied to handle the problem of unknown control directions. Commonly, a Nussbaum-type function satisfies the following properties:

\begin{property} \cite{nussbaum1983} A continuous function $N_i(\chi_i)$ is one choice of the Nussbaum-type functions, if it satisfies
	\begin{align}
		\nonumber	{\lim_{\chi_i \to \infty}}\sup\frac{1}{\chi_i}\int_{0}^{\chi_i}N_{i}(\tau)d\tau=\infty\\
		{\lim_{\chi_i \to \infty}}\inf\frac{1}{\chi_i}\int_{0}^{\chi_i}N_{i}(\tau)d\tau=-\infty.
	\end{align}
\end{property}

\subsection{Problem Formulation}

Consider a multi-agent system containing $N$ first-order linearly parameterized agents with the $i$-th agent being described as follows \cite{renwei2014tac}:
\begin{equation}\label{system}
	\dot{x_i}={\varrho}_{i}{u}_{i}+{\psi}_{i}{\theta}_i
\end{equation}
where $x_i\in\mathbb{R}$ is the state of the agent $i$, $u_i\in\mathbb{R}$ is the control input of the agent $i$, $\rho_i$ is an unknown constant with arbitrary signum which indicates that the unknown control directions of all the agents are non-identical, satisfying $|\rho_i|\in[
\rho_{min},\rho_{max}]$, where $0<\rho_{min}\le\rho_{max}$. $\psi_{i}\in \mathbb{R}^{n}$ represents the regression vector of the linear parameterization terms and $\theta_i$ represents the unknown parameter vector of the linear parameterization terms.

The control objective in this paper is to design the control input $u_i$ for each agent to enable the states of all the agents $x_i, i=1,2,\ldots,N$ achieve consensus, i.e., $\lim\limits_{t\to\infty}|x_i-x_j|=0, i,j=1,2,\ldots,N$ under the undirected and connected communication graph $\mathcal{G}$.
\begin{remark}
	In many practical applications, there exist many kinds of input nonlinearities in the actuators, such as input saturation, input dead zone, backlash-like hysteresis, unknown control directions, etc. As mentioned in \cite{Xiekan2018cyber}, these nonlinearities of the input signals can be modeled as bounded unknown non-zero control coefficients $\rho_i$ with the uncertain signums, i.e., the unknown control directions. It greatly promotes the possibility of practical application of the consensus theory. In particular, different from the overly strict assumptions that $\rho_i$ of all the agents have the identical or partial identical signums in \cite{renwei2014tac}-\cite{LiHYlurenquan2021TNNLSFTC} and \cite{chenci2016Saturated}-\cite{BoFan2019TACpartial}, in this article we assume that all the agents have arbitrary non-identical unknown control directions, i.e., $\rho_i$ with arbitrary signums.
\end{remark}

\begin{lemma}\label{barbalat}
	(Barbalat's Lemma \cite{renwei2014tac}): Considering that  $V(t):\mathbb{R}\to\mathbb{R}$ is a uniformly continuous function for $t>0$, if the limit of the time integral $\lim\limits_{t\to\infty}\int_{0}^{t}V(s)ds$ exists with a finite boundary, then we have $\lim\limits_{t\to\infty}V(t)=0$.
\end{lemma}

\section{Novel Saturated Nussbaum Functions}\label{sec3}

\begin{figure}
	\begin{center}
		\includegraphics[width=0.5\textwidth]{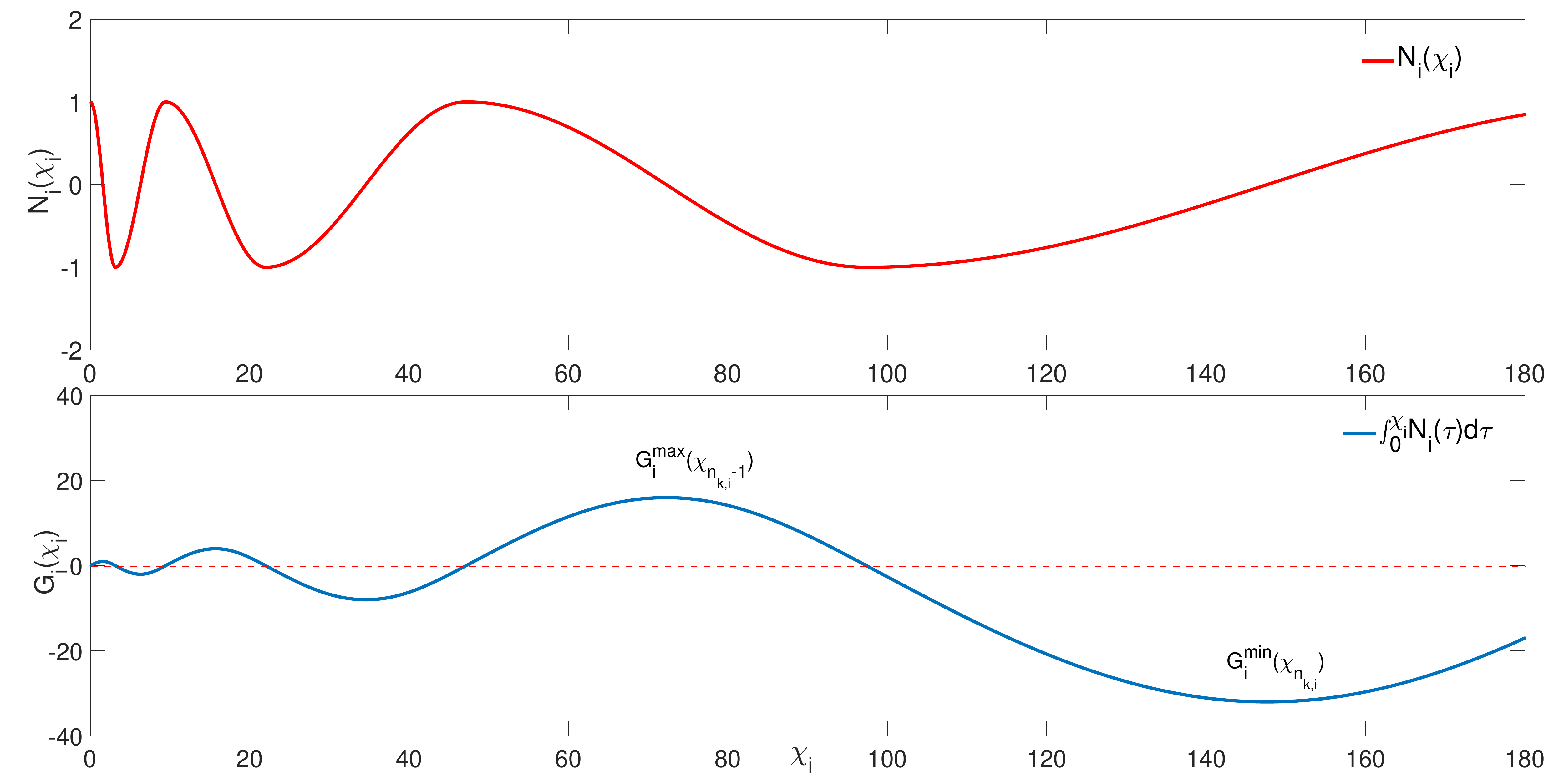}
		\makeatletter\def\@captype{figure}\makeatother
		\caption{An illustration of $N_{i}(\chi_i)$ and $G_{i}(\chi_i)$.}
	\end{center}                       
\end{figure}

Nussbaum function based control scheme has been wildly employed to solve the unknown control directions problem in the consensus problem of multi-agent systems. However, most existing Nussbaum functions are hard to handle the control shock problem and arbitrary unknown control directions in the multi-agent system  simultaneously. The reason is that non-identical directions may cause the offset of the coexisting multiple Nussbaum functions in the same conditional inequality, and exponential amplitude-elongation Nussbaum functions will cause dramatical control shocks in the early stage of control.

To overcome this obstacle, a novel saturated time-elongation Nussbaum function with different frequencies is now proposed as follows:
\begin{small}
	\begin{align}\label{OriNuss}
		\begin{split}
			&N_i(\chi_i)= \\
			&\left \{
			\begin{array}{ll}
				\vdots, & \vdots\\
				{(-1)}^{n_{k,i}-1}a_i\cos[{\frac{1}{{b_{i}}^{n_{k,i}-1}T_{i}}({\chi_i}+\chi_{n_{k,i}-1}})],                                 & (-\chi_{n_{k,i}}, -\chi_{n_{k,i}-1}]\\
				\vdots, & \vdots\\
				-a_i\cos[{\frac{1}{b_{i}T_{i}}({\chi_i}+T_{i}\pi})],                    & (-\pi-b_{i}\pi, -\pi]\\
				a_i\cos{\frac{1}{T_i}{\chi_i}},     & (-\pi, \pi)\\
				-a_i\cos[{\frac{1}{b_{i}T_{i}}({\chi_i}-T_{i}\pi})],                                 & [\pi, \pi+b_{i}\pi)\\
				\vdots, & \vdots\\
				{(-1)}^{n_{k,i}-1}a_i\cos[{\frac{1}{{b_{i}}^{n_{k,i}-1}T_{i}}({\chi_i}-\chi_{n_{k,i}-1}})],                                 & [\chi_{n_{k,i}-1}, \chi_{n_{k,i}})\\
				\vdots, & \vdots
			\end{array}
			\right.
		\end{split}
	\end{align}
\end{small}
where $a_i$, $b_i$, and $T_i$ are positive constants, which satisfy $a_i=\frac{{b_{i+1}^{M-1}}}{\sum_{n=0}^{M-1}{b_{i+1}}^{n}}a_{i+1}$, $b_i={b_{i+1}}^M$, $T_i=(\sum_{n=0}^{M-1}{b_{i+1}}^{n})T_{i+1}$, and $b_i>1$; $n_{k,i}$ is defined as a positive integer which represents the ${n_{k,i}}$-th segment $[\chi_{n_{k,i}-1}, \chi_{n_{k,i}})$ on the picewise Nussbaum function of $i$-th agent; and $\chi_{n_{k,i}}$ is defined as $\chi_{n_{k,i}}=T_{i}\sum_{n=1}^{n_{k,i}}{b_{i}}^{n-1}\pi$. $M\geq4$ is a positive integer.

\begin{lemma}
	The proposed saturated function (\ref{OriNuss}) belongs to Nussbaum functions when it satisfies the \textbf{Property 1}.
\end{lemma}
\begin{proof}
	$N_{i}(\chi_i)$ is an even function. Define $G_i(\chi_i)=\int_{0}^{\chi_i}N_{i}(\tau)d\tau$ as the integration of (\ref{OriNuss}) over the time interval $[0, \chi_i)$, which yields	
	\begin{small}
	\begin{align}\label{IntNuss}
		\begin{split}
			&G_i(\chi_i)= \\
			&\left \{
			\begin{array}{ll}
				\vdots, & \vdots\\
				{(-1)}^{n_{k,i}}a_{i}T_{i}{b_{i}}^{n_{k,i}}\sin[{\frac{1}{{b_{i}}^{n_{k,i}}T_{i}}({\chi_i}+\chi_{n_{k,i}}})],                                 & (-\chi_{n_{k,i}+1}, -\chi_{n_{k,i}}]\\
				\vdots, & \vdots\\
				-a_i\sin[{\frac{1}{b_{i}T_{i}}({\chi_i}+T_{i}\pi})],                    & (-\pi-b_{i}\pi, -\pi]\\
				a_{i}T_i\sin{\frac{1}{T_i}{\chi_i}},     & (-\pi, \pi)\\
				-a_{i}T_{i}b_{i}\sin[{\frac{1}{b_{i}T_{i}}({\chi_i}-T_{i}\pi})],                                 & [\pi, \pi+b_{i}\pi)\\
				\vdots, & \vdots\\
				{(-1)}^{n_{k,i}}a_{i}T_{i}{b_{i}}^{n_{k,i}}\sin[{\frac{1}{{b_{i}}^{n_{k,i}}T_{i}}({\chi_i}-\chi_{n_{k,i}}})],                                 & [\chi_{n_{k,i}}, \chi_{n_{k,i}+1})\\
				\vdots, & \vdots
			\end{array}
			\right.
		\end{split}
	\end{align}
\end{small}
It is obvious that $G_i(\chi_i)$ is an odd function. Over the interval $[\chi_{n_{k,i}-1}, \chi_{n_{k,i}})$, $G_i(\chi_i)$ will reach the maximum value $G_{i}^{max}(\chi_{n_{k,i}})=a_{i}T_{i}{b_{i}}^{n_{k}^{i}-1}$ if $n_{k,i}$ is odd. Similarly, $G_i(\chi_i)$ will reach the minimum value $G_{i}^{min}(\chi_{n_{k,i}})=-a_{i}T_{i}{b_{i}}^{n_{k}^{i}-1}$ if $n_{k,i}$ is even.

Since $b_i$ is greater than 1, it is clear that the following equations are satisfied when $\lim{n_{k,i}}=\infty$ as $\lim{\chi_{i}}=\infty$: 	
\begin{align}
	\nonumber	{\lim_{\chi_i \to \infty}}\sup\frac{1}{\chi_i}\int_{0}^{\chi_i}N_{i}(\tau)d\tau=\frac{a_{i}T_{i}{b_{i}}^{n_{k,i}-1}}{\chi_i}=\infty\\
	{\lim_{\chi_i \to \infty}}\inf\frac{1}{\chi_i}\int_{0}^{\chi_i}N_{i}(\tau)d\tau=-\frac{a_{i}T_{i}{b_{i}}^{n_{k,i}-1}}{\chi_i}=-\infty.
\end{align}
The proof is completed.
\end{proof}

\begin{figure}\label{G}
	\begin{center}
		\includegraphics[width=0.45\textwidth]{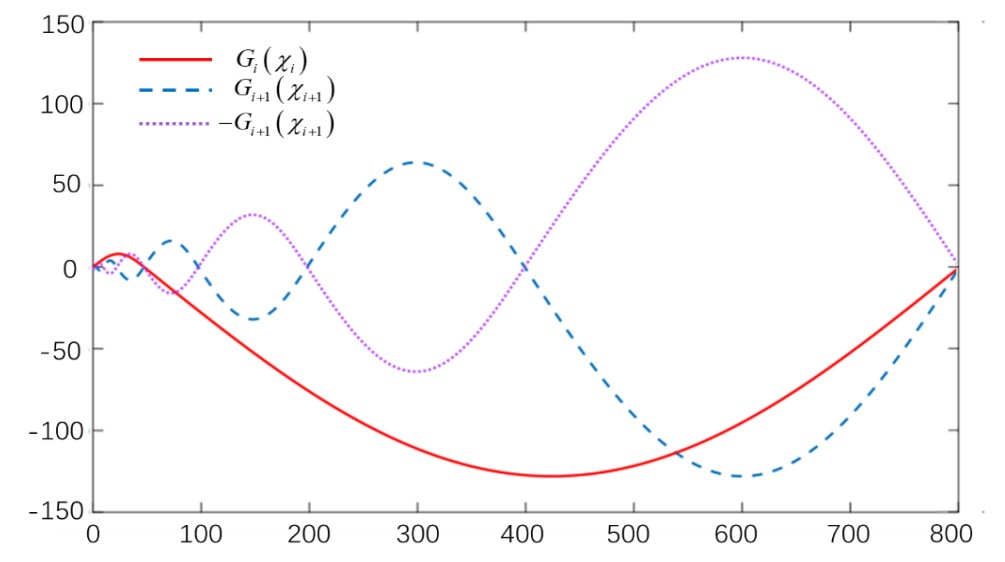}
		\makeatletter\def\@captype{figure}\makeatother
		\caption{$G_i(\chi_{i})$ and sign$(b_{i+1})G_{i+1}(\chi_{i+1})$.}
	\end{center}                   
\end{figure}	

\begin{remark}
	Although the existing Nussbaum functions have separately discussed the fully non-identical UCDs problem and control shock in \cite{HUANG2018auto,chenci2016Saturated,chenci2017tac}, it is still a blank to handle these two issues in the same Nussbaum functions. The novel function (\ref{OriNuss}) proposed in this paper avoids the control shock by transforming the standard amplitude-elongation functions into saturated time-elongation functions, which limits the upper bounds of the Nussbaum gains and the growth speed of control input.  At the same time, Nussbaum functions are assigned different frequencies to avoid the risk of mutual cancellation. 
\end{remark}

\begin{remark}
	In the novel proposed Nussbaum function, the half-period of the $\cos$ functions are employed as the piecewise original functions instead of the $\sin$ functions in \cite{chenci2016Saturated}, such that the forms of integrated functions $G_i(\chi_{i})$ are the half-period of the $\sin$ functions. This type of saturated Nussbaum function is beneficial to construct contradictions and prove Lemma \ref{0.5} and Theorem \ref{theorem}. Besides, the assumption of identical UCDs in \cite{chenci2016Saturated} is eliminated.
\end{remark}

\begin{remark}
	With the constant $M$ satisfying $M\geq4$, each Nussbaum function is endowed different frequencies which means that each segment of the $i$-th agent is divided into $M$ segments for the  ${(i+1)}$-th agent. Different from $G_{i}(k_i)=e^{\alpha k_{i}}\sin(M^{i}k_{i}-\epsilon_{i})$ in \cite{HUANG2018auto}, the Nussbaum functions designed as a time-elongation form in this paper can avoid the fast increase of Nussbaum gain and smooth the control input in the early stage of parameter approximation. 
\end{remark}

With the parameters $a_i$, $b_i$, and $T_i$ satisfying the propotion relationship between $G_i$ and $G_{i+1}$ in (\ref{IntNuss}), i.e., max$|G_{i}(\chi_{i})|=$max$|G_{i+1}(\chi_{i+1})|$,  $\forall \chi_{i}\in[\chi_{n_{k,i}-1}, \chi_{n_{k,i}}), \chi_{i+1}\in[\chi_{M\cdot n_{k,i}-1}, \chi_{M\cdot n_{k,i}})$, where $\chi_{M\cdot n_{k,i}}=\chi_{n_{k,i+1}}$, we have the following Lemmas and Corollary. For the convenience of parameter selection, the following discussion in this paper is based on the chosen parameter $M=4$. $M>4$ does not affect the correctness of the following Lemmas and Corollary.

\begin{lemma}\label{0.5}
	With $M=4$, there exist intervals $[\chi_{n_{k,i+1}}^{in}, \chi_{n_{k,i+1}}^{out}],  k=1,\ldots,\infty$, on which  $\forall \chi_{i}\in[\chi_{n_{k,i+1}}^{in}, \chi_{n_{k,i+1}}^{out}]$, $sign({\varrho}_i)G_i({\chi_i}) \le0.5G_{i+1}^{min}(\chi_{i+1})<0$, where $\chi_{n_{k,i+1}}^{in}$ and $\chi_{n_{k,i+1}}^{out}$ represent the left end point and right end point of the interval over the $k$-th segment on the Nussbaum function of $(i+1)$-th agent. 
\end{lemma}

\begin{proof}
	The details are presented in \textbf{Appendix A}.
\end{proof}

\begin{remark}
	With Lemma \ref{0.5}, we can always find an interval on the Nussbaum function of $(i+1)$th agent, on which $sign(\varrho_i)G_{i}$ and  $sign(\varrho_{i+1})G_{i+1}$ will be both smaller than $0.5G_{i+1}^{min}(\chi_{i+1})$ whether they are in the same direction or not. This provides an upper bound for the minimum value of the Nussbaum functions' integrations for all agents with arbitrary directions in the following proof of Theorem \ref{theorem}.

\end{remark}

\begin{corollary}\label{coro}
	As a simple extension of Lemma \ref{0.5}, when there are $L$ agents, there must be an interval $[\chi_{n_{k,L}}^{in}, \chi_{n_{k,L}}^{out}],  k=1,\ldots,\infty$ such that $\forall \chi_{i} \in[\chi_{n_{k,L}}^{in}, \chi_{n_{k,L}}^{out}]$, $sign({\varrho}_i)G_i({\chi_i}) \le sign({\varrho}_{L})G_{L}({\chi_{L}})\le0.5G_{L}^{min}(\chi_{L})<0$, where $\chi_{n_{k,L}}^{in}$ and $\chi_{n_{k,L}}^{out}$ represent the inner point and outer point of the interval over the $k$-th segment on the function of $L$-th agent, i.e., the function with the fastest frequency in a set of agents. The proof is omitted here.
\end{corollary}

\begin{lemma}\label{fangsuo}
	For $\chi_{i} \le \chi_{n_{k,L}}^{in}$, the inequality  $G_{i}(\chi_i)\le G_{L}^{max}(\chi_{n_{k,L}-1}),  i=1,\ldots,L$ holds.
\end{lemma}

\begin{proof}
	
	First consider the agents $i$ and $i+1$, suppose that the interval $[\chi_{n_{k,i+1}}^{in}, \chi_{n_{k,i+1}}^{out}]\in [\chi_{n_{k*,i}-1}, \chi_{n_{k*,i}})$, it is clear that $M(n_{k*,i}-1)+1\le n_{k,i+1}-1$. For $\chi_{i} \le \chi_{n_{k,i+1}}^{in}$, with the propotion relationship between $G_{i}$ and $G_{i+1}$, we have $G_{i}(\chi_{i})\le G_{i}^{max}(\chi_{n_{k*,i}-1})=G_{i+1}^{max}(\chi_{M({n_{k*,i}-1})})\le G_{i+1}^{max}(\chi_{n_{k,i+1}-1})$. Similarly, if there exist $L$ agents, the above inequality also holds, i.e., $G_{i}(\chi_i)\le G_{L}^{max}(\chi_{n_{k,L}-1}), i=1,\ldots,L$.

	
\end{proof}

\begin{theorem}\label{theorem}
	Let all of $V_{i}(t)$ and $\chi_{i}(t)$ for $i=1, \ldots, N$ be smooth functions with the initial $\chi_{i}(0)$ being bounded and $V_{i}(t)$ being non-negative on the time interval $[0, t_f)$. The saturated Nussbaum function $N_{i}(\chi_{i})$ is constructed as (\ref{OriNuss}) with constants $a_N$ and $b_{N}$ satisfying
	\begin{align}
		\nonumber  			a_N>\max\{\frac{2\bar{\eta}\pi}{\varrho_{min}(b_N-1)}, -\frac{\bar{\eta}\pi\xi}{\kappa}\},  \quad b_{N}>2(N-1)\frac{\varrho_{max}}{\varrho_{min}},
	\end{align}	
	where $\kappa=\frac{(N-1)\varrho_{max}}{b_N}-\frac{\varrho_{min}}{2}$ and $\xi=\frac{N}{b_N-1}+\frac{N+4}{6}$.

	If the following inequality holds:	
	\begin{equation}
		V(t)\le\sum_{i=1}^{N}\int_{0}^{t}{\varrho_{i}N_{i}(\chi_{i})\dot{\chi_{i}}(\tau)d\tau}+\sum_{i=1}^{N}\int_{0}^{t}\eta_{i}\dot{\chi_{i}}(\tau)d\tau+c, \forall t \in [0,t_f)
	\end{equation}
	where $\eta_{i}$ and $c$ are constants with $\eta_{i}>0$, $\varrho_{i}$ could be any signs with $|\varrho_{i}|\in [\varrho_{min}, \varrho_{max}], i=1,2,\ldots,N$, then, $V(t), \chi_{i}(t)$, and $\sum_{i=1}^{N}\int_{0}^{t}{\varrho_{i}N_{i}(\chi_{i})\dot{\chi_{i}}(\tau)d\tau}$ are all bounded on $[0,t_f)$.
\end{theorem}

\begin{proof}
	
	From (6) we obtain
	\begin{align}
		\nonumber	V(t)\le&\sum_{i=1}^{N}\int_{0}^{t}{\varrho_{i}N_{i}(\chi_{i})\dot{\chi_{i}}(\tau)d\tau}+\sum_{i=1}^{N}\int_{0}^{t}\eta_{i}\dot{\chi_{i}}(\tau)d\tau+c\\
		=&\sum_{i=1}^{N}\int_{0}^{\chi_{i}(t)}{\varrho_{i}N_{i}(\sigma)d\sigma}+\sum_{i=1}^{N}\eta_{i}\chi_{i}(t)+\bar{c}\label{Vchange}
	\end{align}
	where $\bar{c}=\sum_{i=1}^{N}\int_{\chi_{i}(0)}^{0}{\varrho_{i}N_{i}(\sigma)d\sigma}-\sum_{i=1}^{N}\eta_{i}\chi_{i}(0)+c$ is a constant. Then, similar to \cite{renwei2014tac}, we will discuss the boundness of all $\chi_{i}(t)$ by seeking a contradiction. Assume that $\chi_{i}(t), i=1, \ldots, L$ are unbounded, while $\chi_{i}(t), i=L+1, \ldots, N$ are bounded, $1\le L\le N$.  Considering the $N_{i}(\chi_{i})$ is odd, we investigate the situation of $\chi_{i}>0$. The proof that $\chi_{i}<0$ is similar and omitted here. Then the inequation (\ref{Vchange}) can be written as 	
	\begin{align}
		\nonumber	V(t)\le&\sum_{i=1}^{N}\int_{0}^{\chi_{i}(t)}{\varrho_{i}N_{i}(\sigma)d\sigma}+\sum_{i=1}^{N}\eta_{i}\chi_{i}(t)+\bar{c}\\
		\nonumber		=&\sum_{i=1}^{L}\int_{0}^{\chi_{i}(t)}{\varrho_{i}N_{i}(\sigma)d\sigma}+\sum_{i=1}^{L}\eta_{i}\chi_{i}(t)\\
		\nonumber		&+\sum_{i=L+1}^{N}\int_{0}^{\chi_{i}(t)}{\varrho_{i}N_{i}(\sigma)d\sigma}+\sum_{i=L+1}^{N}\eta_{i}\chi_{i}(t)+\bar{c}\\
		=&\sum_{i=1}^{L}\int_{0}^{\chi_{i}(t)}{\varrho_{i}N_{i}(\sigma)d\sigma}+\sum_{i=1}^{L}\eta_{i}\chi_{i}(t)+\tilde{c},\label{Vsep}
	\end{align}
	where $\tilde{c}=\sum_{i=L+1}^{N}\int_{0}^{\chi_{i}(t)}{\varrho_{i}N_{i}(\sigma)d\sigma}+\sum_{i=L+1}^{N}\eta_{i}\chi_{i}(t)+\bar{c}$ is bounded.

	On the time interval $[0,t_f)$, let $\chi_{i*}=max\{\chi_{1}, \ldots, \chi_{L}\}$. With Lemma \ref{0.5} and Corollary \ref{coro}, we can always find an interval $[\chi_{n_{k,L}}^{in}, \chi_{n_{k,L}}^{out}],  k=1,\ldots,\infty$ such that $sign(\varrho_{L}){(-1)}^{n_{k,L}-1}\sin[{\frac{1}{{b_{L}}^{n_{k,L}-1}T_{L}}({\chi_L}-\chi_{n_{k,L}-1})}]\le-0.5$.  Clearly $\chi_{n_{k,L}}^{in}$ and $\chi_{n_{k,L}}^{out}$ are multiple and unbounded as $k \to \infty$. Assume that $\chi_{i*}$ has entered this interval. Then the analysis can be divided into two parts, i.e., $sign(\varrho_{1})=1$ and $sign(\varrho_{1})=-1$.

	\textit{	\textbf{Case One}: $sign(\varrho_{1})=1.$ }

	For this case, the analysis also has two parts, i.e., all $\chi_{i}$ have been on the interval $[\chi_{n_{k,L}}^{in}, \chi_{n_{k,L}}^{out}],  k=1,\ldots,\infty$, and only part of $\chi_{i}$ have been on this interval.
	
	(i) If all $\chi_{i}$ have been on the interval $[\chi_{n_{k,L}}^{in}, \chi_{n_{k,L}}^{out}],  k=1,\ldots,\infty$, the inequation (\ref{Vsep}) can be written as 
	\begin{align}
		\nonumber	V(t)\le&\sum_{i=1}^{L}\int_{0}^{\chi_{i}(t)}{\varrho_{i}N_{i}(\sigma)d\sigma}+\sum_{i=1}^{L}\eta_{i}\chi_{i}(t)+\tilde{c}\\
		\nonumber	\le&L(0.5\varrho_{min}G_{L}^{min}(\chi_{n_{k,L}})+\bar{\eta}\chi_{n_{k,L}}^{out})+\tilde{c}\\
		\nonumber	=&L(-0.5\varrho_{min}a_{L}T_{L}{b_L}^{n_{k,L}-1}+\bar{\eta}\chi_{n_{k,L}})+\tilde{c}\\
	\nonumber	=&LT_{L}{b_L}^{n_{k,L}}(-0.5\varrho_{min}\frac{a_{L}}{b_L}+\bar{\eta}\pi\frac{\frac{1}{{b_L}^{n_{k,L}}}-1}{1-b_L}	)+\tilde{c}\\
		=&LT_{L}{b_L}^{n_{k,L}}(-0.5\varrho_{min}\frac{a_{L}}{b_L}+\bar{\eta}\pi\frac{1}{{b_L}-1}	)+\tilde{c}.
	\end{align}	
	Since $a_i$ and $b_{i}$ are chosen as $a_{i}>\frac{2\bar{\eta}\pi b_{i}}{\varrho_{min}(b_i-1)}$ and $b_i>1$, it is obvious that $V(t)\to -\infty$ as $n_{k,L}\to \infty$, which yields a contradiction.

	(ii) If only part of $\chi_{i}$ have been on the interval. Without loss of generality, we consider only one agent entering the interval $[\chi_{n_{k,L}}^{in}, \chi_{n_{k,L}}^{out}]$, i.e., $\chi_{i*}\in [\chi_{n_{k,L}}^{in}, \chi_{n_{k,L}}^{out}]$ and $\chi_{i}<\chi_{n_{k,L}}^{in}, 1\le i \le L \quad and \quad i\neq i*$.

	With Lemmas \ref{0.5}, \ref{fangsuo}, and Corollary \ref{coro}, for the agent $i \in [1,L]$, regardless of whether $sign(\varrho_i)=1$ or $sign(\varrho_i)=-1$, with $\chi_{i*}$ entering the interval  $[\chi_{n_{k,L}}^{in}, \chi_{n_{k,L}}^{out}]$ and other $\chi_i$ being smaller than $\chi_{n_{k,L}}^{in}$, the following inequations are satisfied:
	\begin{align}
		\nonumber  	G_{i*} \le 0.5G_{L}^{min}(\chi_{n_{k,L}}),\\
		\nonumber  	G_{i, i\neq i*}\le G_{L}^{max}(\chi_{n_{k,L}-1}).
	\end{align}	
	Consequently, the inequation (\ref{Vsep}) can be written as:
	\begin{align}
		\nonumber  V(t)\le&\sum_{i=1}^{L}\int_{0}^{\chi_{i}(t)}{\varrho_{i}N_{i}(\sigma)d\sigma}+\sum_{i=1}^{L}\eta_{i}\chi_{i}(t)+\tilde{c}\\
		\nonumber  	=&\sum_{i=1, i\neq i*}^{L}\int_{0}^{\chi_{i}(t)}{\varrho_{i}N_{i}(\sigma)d\sigma}+\int_{0}^{\chi_{i*}(t)}{\varrho_{i*}N_{i*}(\sigma)d\sigma}\\	
		\nonumber  	&+\sum_{i=1, i\neq i*}^{L}\eta_{i}\chi_{i}+\eta_{i*}\chi_{i*}+\tilde{c}\\	
		\nonumber  	\le&(L-1)\varrho_{max}G_{L}^{max}(\chi_{n_{k,L}-1})+0.5\varrho_{min}G_{L}^{min}(\chi_{n_{k,L}})\\
		\nonumber &+\bar{\eta}(L-1)\chi_{n_{k,L}}^{in}+\bar{\eta}\chi_{n_{k,L}}^{out}+\tilde{c}\\	
		\nonumber  	=&G_{L}^{max}(\chi_{n_{k,L}-1})[(L-1)\varrho_{max}-0.5\varrho_{min}b_{L}]+\bar{\eta}L\chi_{n_{k,L}-1}\\	
		\nonumber  	&+\bar{\eta}(\frac{L}{6}+\frac{2}{3})(\chi_{n_{k,L}}-\chi_{n_{k,L}-1})+\tilde{c}\\	
		\nonumber  	=&T_{L}{a_L}{b_{L}}^{n_{k,L}-2}[(L-1)\varrho_{max}-\frac{\varrho_{min}}{2}b_{L}]+\\
		\nonumber  	&\bar{\eta}LT_{L}\pi\frac{1-b_{L}^{n_{k,L}-1}}{1-b_L}+\bar{\eta}(\frac{L}{6}+\frac{2}{3})T_{L}\pi b_{L}^{n_{k,L}-1}+\tilde{c}\\	
		\nonumber  	=&T_{L}a_L{b_{L}}^{n_{k,L}-1}[\frac{(L-1)\varrho_{max}}{b_L}-\frac{\varrho_{min}}{2}\\
		\nonumber&+\frac{\bar{\eta}\pi}{a_L}(\frac{L}{b_L-1}+\frac{L+4}{6})]+\tilde{c}\\
		=&T_{L}a_L{b_{L}}^{n_{k,L}-1}(\kappa+\frac{\bar{\eta}\pi\xi}{a_L})+\tilde{c},
	\end{align}
	where $\kappa=\frac{(L-1)\varrho_{max}}{b_L}-\frac{\varrho_{min}}{2}$ and $\xi=\frac{L}{b_L-1}+\frac{L+4}{6}$. Since $a_N$ and $b_{N}$ are chosen as $a_{N}>-\frac{\bar{\eta}\pi\xi}{\kappa}$ and $b_N>\frac{2(N-1)\varrho_{max}}{\varrho_{min}}$, it is obvious that $V(t)\to -\infty$ as $n_{k,L}\to \infty$, which yields a contradiction.
	
	Therefore, due to the above discussion, $V(t), \chi_{i}(t)$, and $\sum_{i=1}^{N}\int_{0}^{t}{\varrho_{i}N_{i}(\chi_{i})\dot{\chi_{i}}(\tau)d\tau}$ are all bounded on $[0,t_f)$.

\textit{	\textbf{Case Two}: $sign(\varrho_{1})=-1.$  }

The proof is similar to Case One and omitted here. The only difference is that the intervals $[\chi_{n_{k,L}}^{in}, \chi_{n_{k,L}}^{out}]$ will move backward by one period.

The proof is completed.

\end{proof}

\begin{remark}
	In this paper, we establish a special proportional relationship  for the $i$-th and $(i+1)$-th Nussbaum functions to construct contradictions by setting the parameters $a_i$, $b_i$, and $T_i$. With the special relationship, the maximum and minimum of $G_i$ can be scaled to two adjacent maximum and minimum values on the same Nussbaum function, i.e., the agent $L$. Different from existing papers, this method is greatly convenient for the analysis of the multiple Nussbaum functions in one equality of $V(t)$ by determining the interval $[\chi_{n_{k,L}}^{in}, \chi_{n_{k,L}}^{out}]$. The fully non-identical UCDs are also achieved in our analysis framework.
\end{remark}

\section{Control Design for Multi-Agent Systems with Fully Non-identical Unknown Control Directions}\label{sec4}

In this section, a Nussbaum-type function-based adaptive distributed control algorithm is proposed for the consensus problem of the first-order linearly parameterized multi-agent system (2). The control input $u_i$ is constructed as follows
\begin{align}\label{controlinput}
	u_i=&-N_i(\chi_i)u_{Ni},
\end{align}
where $N_i(\chi_i)$ is the Nussbaum-type gain designed as (3) for each agent associated with the value of $\chi_i$, and $u_{Ni}$ is an auxiliary controller designed as 
\begin{align}\label{auxinput}
	u_{Ni}=\sum_{i=1}^{N}a_{ij}(x_i-x_j)+\psi_{i}^T\hat{\theta_i},
\end{align}
where $\hat{\theta_i}$ is the estimated value of the parameter vector $\theta_i$. The adaptive updating laws for $\hat{\theta_i}$ and $\chi_{i}$ are designed as 
\begin{align}\label{paraupdate}
		\dot{\hat{\theta_i}}=&\zeta_i\psi_i[\sum_{i=1}^{N}a_{ij}(x_i-x_j)],\\
	\nonumber\dot{\chi_i}=&\gamma_i[\sum_{i=1}^{N}a_{ij}(x_i-x_j)]\times [\sum_{i=1}^{N}a_{ij}(x_i-x_j)+\psi_{i}^T\hat{\theta_i}],
\end{align}
with $\zeta_i$ and $\gamma_{i}$ being positive constants. The initial values $\hat{\theta_i}(0)$ and $\chi_{i}(0)$ are chosen arbitrary.

\begin{theorem}
	Consider the undirected graph $\mathcal{G}$ and
	the multi-agent system (\ref{system}), with the distributed adaptive control laws (\ref{controlinput}), (\ref{auxinput}), and parameters updating laws (\ref{paraupdate}), then the
	states of all agents are guaranteed to be bounded and asymptotically
	reach consensus.
\end{theorem}

\begin{proof}
	Denote $x=[x_1,x_2,\dots,x_N]^T$ and $\tilde{\theta}=[\tilde{\theta}_1,\tilde{\theta}_2,\ldots,\tilde{\theta}_N]^T$, where $\tilde{\theta}_i=\hat{\theta_i}-\theta$. Note that $\mathcal{G}$ is connected and undirected, we have $e=Lx$, where $e=[e_1,e_2,\ldots,e_N]^T$ and $e_i=\sum_{i=1}^{N}a_{ij}(x_i-x_j)$ represents the combined error of agent $i$. 
	
	Consider the following Lyapunov function candidate for the system (\ref{system})	
	\begin{equation}\label{Lya}
		V=\frac{1}{2}x^TLx+\frac{1}{2}\tilde{\theta}^TH\tilde{\theta}
	\end{equation}
	where $H=[diag(\zeta_1,\ldots,\zeta_N)]^{-1}$ and Laplace matrix $L$ are both symmetric positive definite.

	Taking the time derivative of (\ref{Lya}) along (\ref{paraupdate}) yields \begin{small} \begin{align}\label{Lyadot}
		\nonumber	\dot{V}=&x^TL\dot{x}+\tilde{\theta}^TH\dot{\tilde{\theta}}\\
		\nonumber	=&-\sum_{i=1}^{N}e_{i}^2+\sum_{i=1}^{N}e_{i}^2-\sum_{i=1}^{N}e_i(\varrho_iN_i(\chi_i)e_i+\psi_{i}^T\hat{\theta_i})+\sum_{i=1}^{N}\tilde{\theta_i}\psi_{i}e_i\\
		\nonumber	=&-\sum_{i=1}^{N}e_{i}^2-\sum_{i=1}^{N}e_{i}(\varrho_iN_i(\chi_i)+1)u_{Ni}\\
		=&-\sum_{i=1}^{N}e_{i}^2-\sum_{i=1}^{N}\frac{\varrho_i}{\gamma_{i}}N_i(\chi_i)\dot{\chi_i}+\sum_{i=1}^{N}\frac{1}{\gamma_i}\dot{\chi_i}.
	\end{align}
\end{small}
	Integrating both sides of (\ref{Lyadot}) from 0 to $t$ yields	
	\begin{align}
		\nonumber	V(t)=&-\sum_{i=1}^{N}\int_{0}^{t}e_i^2(\tau)d\tau-\sum_{i=1}^{N}\int_{0}^{t}\frac{\varrho_i}{\gamma_{i}}N_i(\chi_i(\tau))\dot{\chi_i}(\tau)d\tau\\
		&+\sum_{i=1}^{N}\int_{0}^{t}\frac{1}{\gamma_i}\dot{\chi_i}(\tau)d\tau+\Delta,
	\end{align}
	where $\Delta$ is a bounded constant.
	
	Therefore, by applying Theorem \ref{theorem}, we can obtain that $V(t)$, $\chi_{i}(t)$,  and  $\sum_{i=1}^{N}\int_{0}^{t}{\frac{\varrho_{i}}{\gamma_{i}}N_{i}(\chi_{i})\dot{\chi_{i}}(\tau)d\tau}$ must be all bounded on $[0,t_f)$. As $t_f=\infty$, the system has no potential for finite-time escape  and $e_i^2(t)$ is integrable on $[0,\infty)$ \cite{renwei2014tac}. Consequently, based on Barbalat's Lemma \ref{barbalat}, we can obtain that $\lim_{t \to \infty}e_i^2(t)=0, 1\le i\le N$. Since $\mathcal{G}$ is undirected and connected, it implies that $\lim_{t \to \infty}(x_i(t)-x_j(t)) = 0$ as $\lim_{t \to \infty}e_i^2(t)=0$. 
	
	The proof is completed.
\end{proof}

\begin{figure}[!t]
	\begin{center}
		\includegraphics[width=0.2\textwidth]{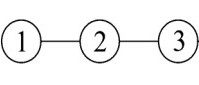}
		\caption{Communication topology for a group of three agents.}\label{topology}
	\end{center}   
\end{figure}

\section{Simulation}\label{sec5}

\begin{figure*}[!t]
	\centering
	\subfloat[]{\label{novelu}
		\includegraphics[width=0.5\textwidth]{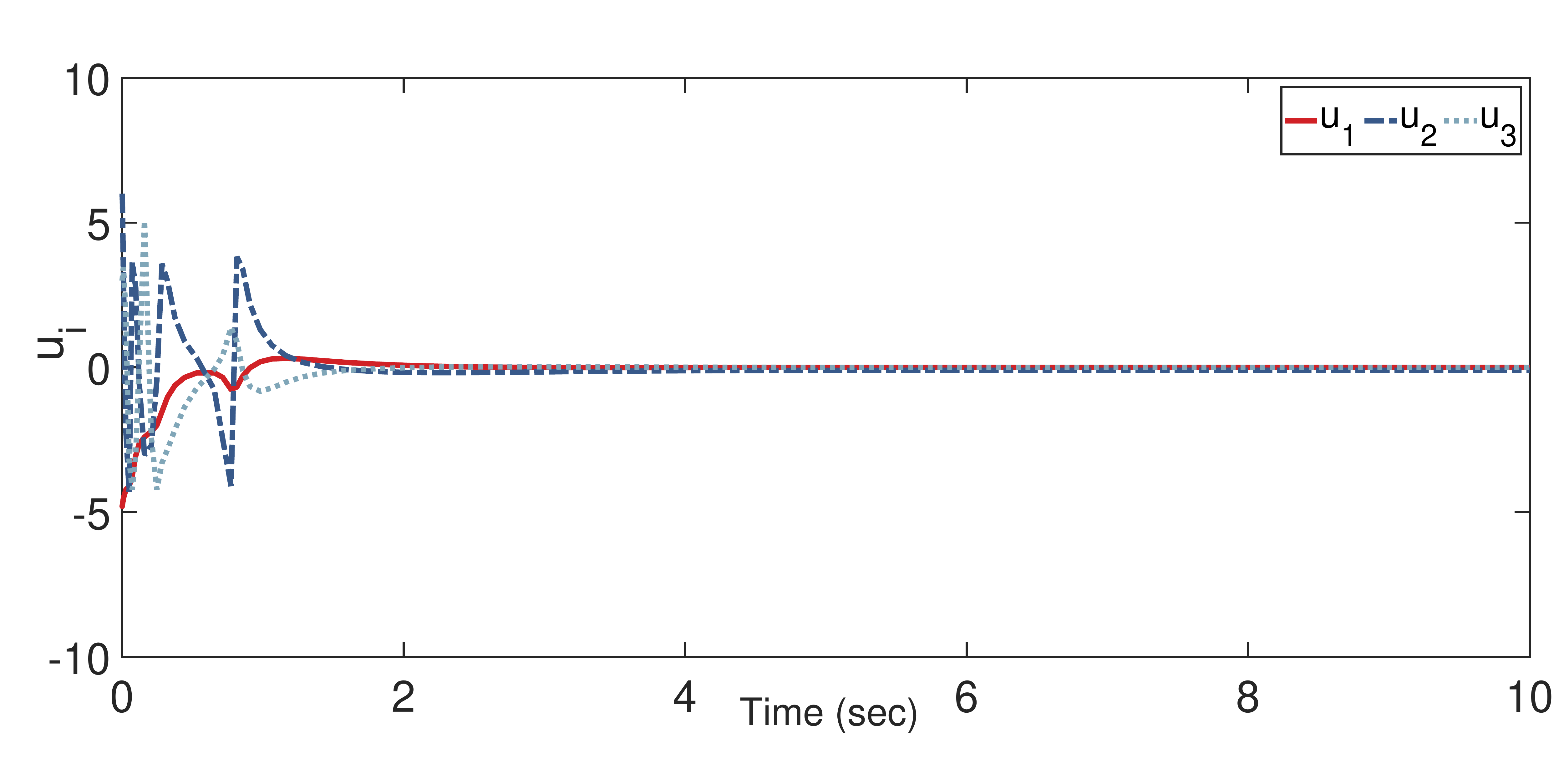}}
	\subfloat[]{\label{novelx}
		\includegraphics[width=0.5\textwidth]{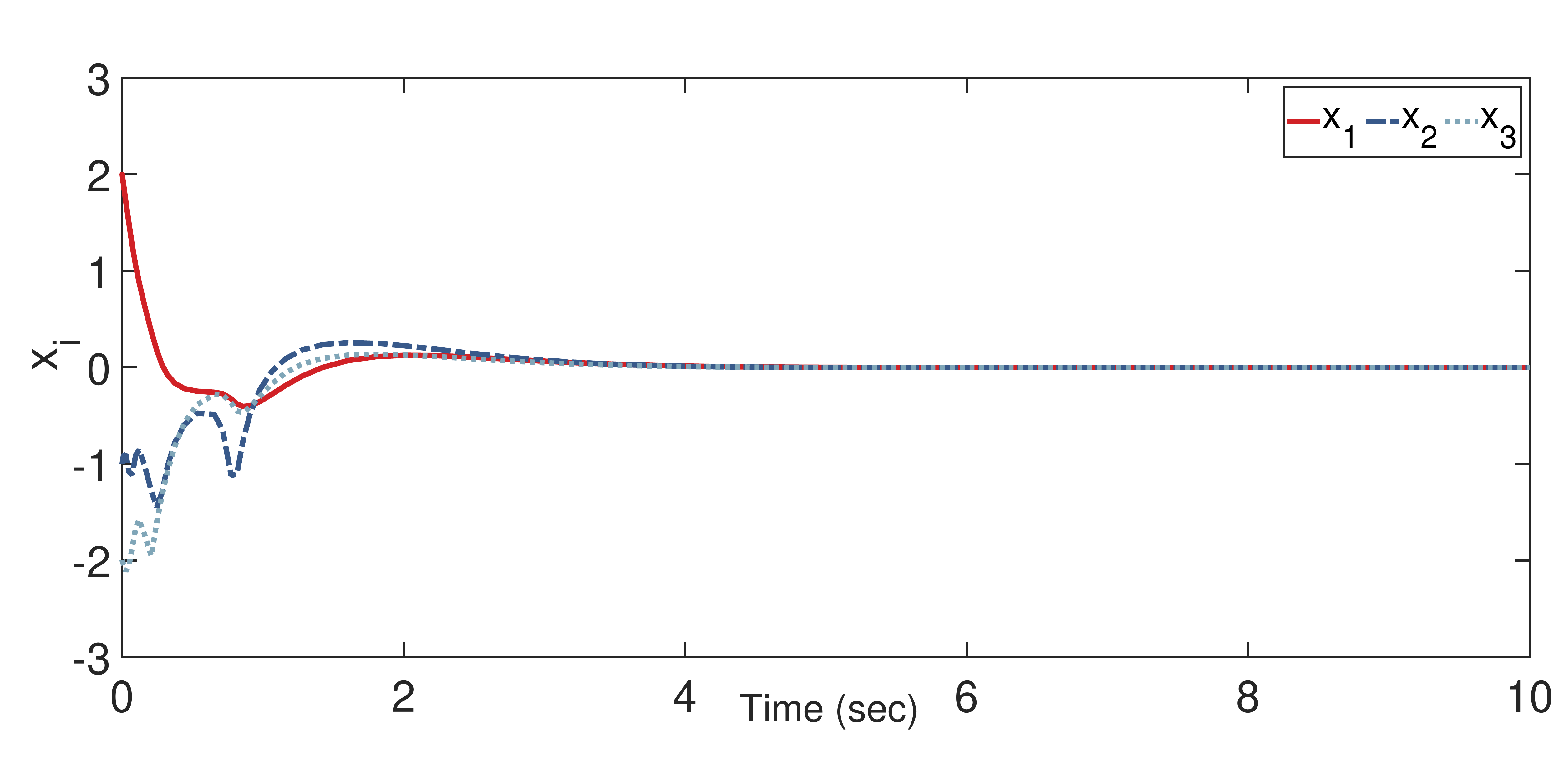}}	
	
	\subfloat[]{\label{novelchi}
		\includegraphics[width=0.5\linewidth]{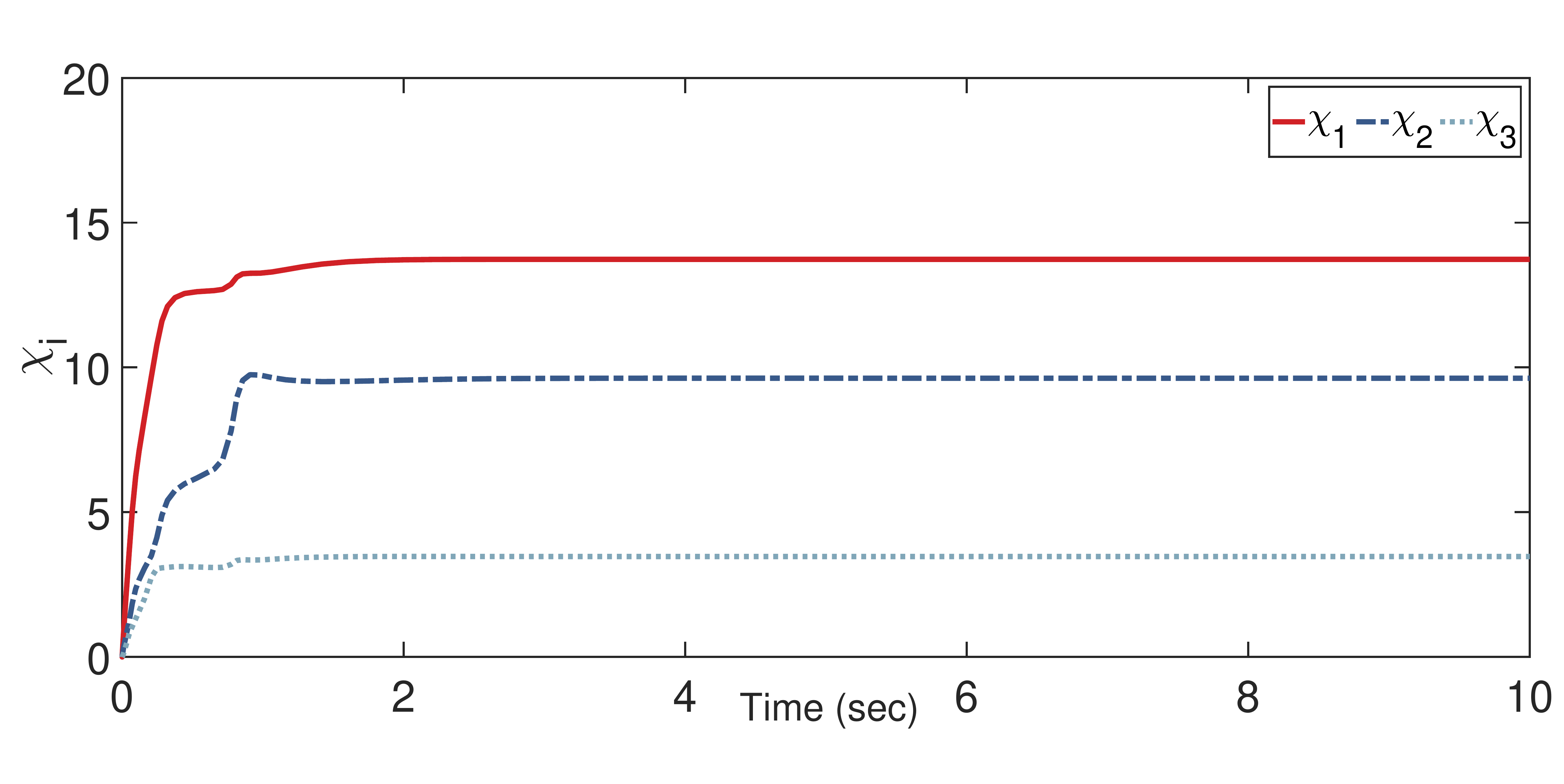}}
	\subfloat[]{\label{novelN}
		\includegraphics[width=0.5\textwidth]{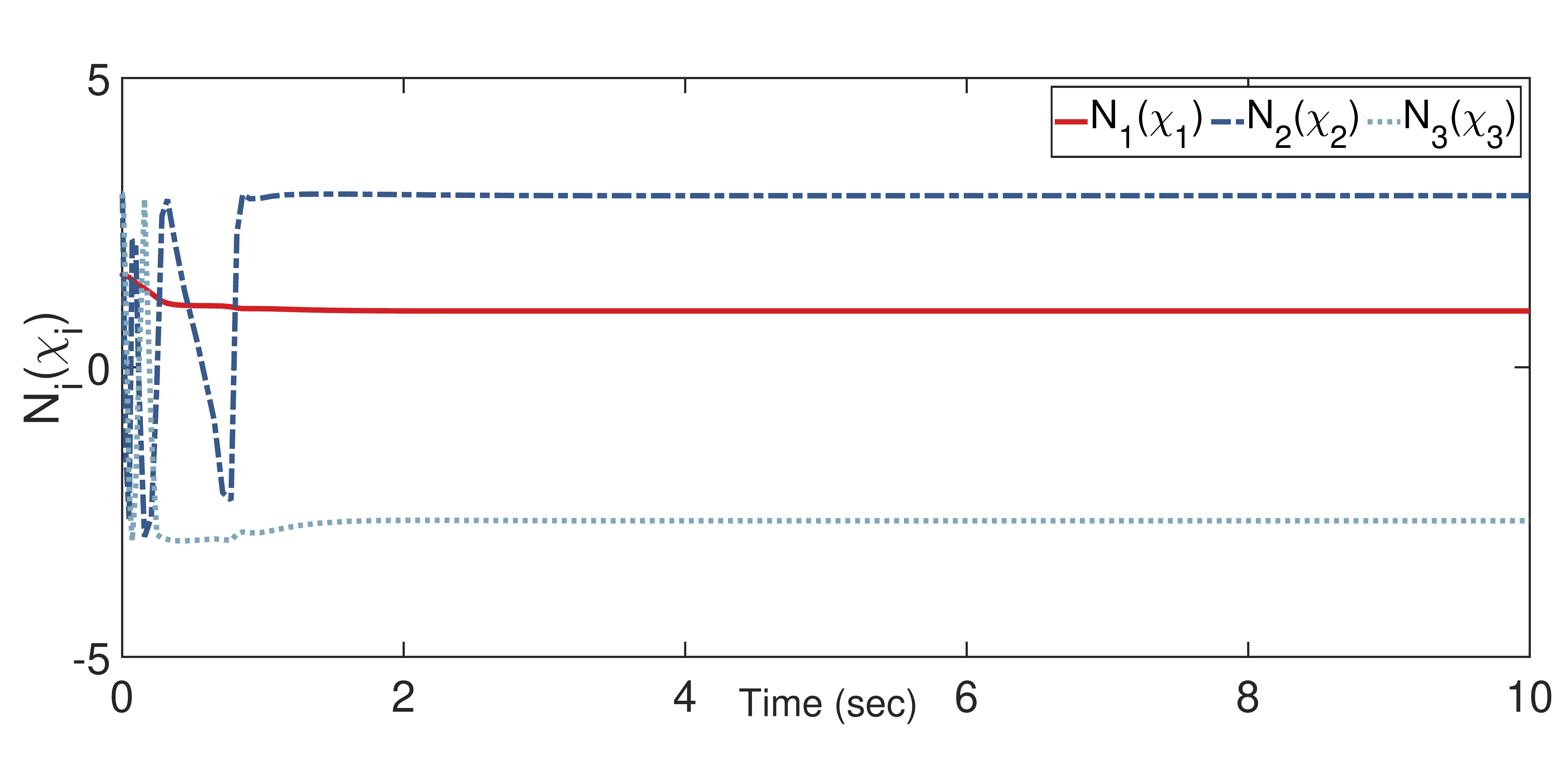}}
	\caption{Results by the novelly proposed Nussbaum based approach. \protect\subref{novelu} Control input $u_i$, \protect\subref{novelx} State $x_i$, \protect\subref{novelchi} Updating gain $\chi_i$, \protect\subref{novelN} Saturated Nussbaum function $N_i(\chi_i)$.}\label{novel}
\end{figure*}

In this section, a numerical simulation of the consensus of the system (\ref{system}) with 3 agents is provided to illustrate the proposed control algorithm, as well as a comparative study between the traditional Nussbaum functions  \cite{HUANG2018auto} based controller and the novel proposed saturated Nussbaum functions (\ref{OriNuss}) based controller in this paper. The communication topology among agents is shown in Fig. \ref{topology}. To facilitate the comparison, the initial values of the system and the gain coefficients of the controller in both control algorithms are selected the same. The only differences are the Nussbaum functions and associated parameters.

Here we consider the same linearly parameterized system (\ref{system}) as that in the paper \cite{renwei2014tac}, in which the regression vectors are taken as $\psi_1(x_1)=\sin(x_1)$, $\psi_2(x_2)=\cos(x_2)$, and $\psi_3(x_3)=x_3$; the parameter vectors are set as $\theta_1=-1.1, \theta_2=0.2$, and $\theta_3=-0.6$; the gains in the updating law (\ref{paraupdate}) for $\hat{\theta_i}$ are set as $\zeta_1=0.6, \zeta_2=1.6$, and $\zeta_3=2.3$; the initial conditions of $x_i$ are set as $x_1=2, x_2=-1$ and $x_3=-2$; the initial value of $\theta_i$ is set as zero. Differently, the unknown control directions are assumed to be arbitrary directions instead of the identical directions in \cite{renwei2014tac}. The unknown coefficients of control input in this paper are specified as $\rho_1=2.2, \rho_2=2$, and $\rho_3=-1.8$. Then we will show the effectiveness and comparison of control algorithms with novel saturated Nussbaum functions (\ref{OriNuss}) and the traditional Nussbaum functions in \cite{HUANG2018auto}.

\textit{Consensus algorithms with novel Nussbaum functions:}
The control parameters are set as $b_1=3^{16}, b_2=3^4$, and $b_3=3$; $a_1={a_2}\times\frac{3^{16}-3^4}{3^{16}-1}, a_2={a_3}\times\frac{27}{40}$, and $a_3=3$; $T_1=T_2\times\frac{1}{1+3^4+3^8+3^{16}},T_2=T_3\times\frac{1}{40}$, and $T_3=100$; $\gamma_{i}=10$. Fig. \ref{novel} shows the simulation results of consensus algorithms with novel Nussbaum functions. It can be witnessed that all agents achieve consensus in 3 seconds, and $\chi_i$ and $N_i(\chi_i)$ are all bounded which consistent with the Theorem \ref{theorem}. It is worth pointing that the control input $u_i$ is bounded, as well as the value of Nussbaum gain $N_i(\chi_i)$ is also bounded in the interval $[-a_3, a_3]$.

\textit{Consensus algorithms with Nussbaum functions in \cite{HUANG2018auto}:}
Following the Nussbaum functions in  \cite{HUANG2018auto}, i.e. $N_i(\chi_i)=\frac{\alpha^2\beta^{2i}+1}{\beta^i\sqrt{\alpha^2\beta^{2i}+1}}e^{\alpha|\chi_i|}\sin(\frac{\chi_i}{\beta^i})$, the parameters in this compared Nussbaum functions are chosen as $\gamma_i=5$, $\alpha=4$, and $\beta=0.2$. Fig. \ref{standard} shows the simulation results of consensus algorithms with traditional Nussbaum functions \cite{HUANG2018auto}. It can be witnessed that all agents achieve consensus in 2 seconds, which is faster than the novel Nussbaum-based consensus algorithms. Note that the control input $u_i$  quickly stabilizes to a constant value but there is a dramatical control shock at the begining period. The value of Nussbaum gain $N_i(\chi_i)$ is also bounded in a much wider interval than $[-a_3, a_3]$.

\begin{figure*}[!t]
	\centering
	\subfloat[]{\label{huangu}
		\includegraphics[width=0.5\textwidth]{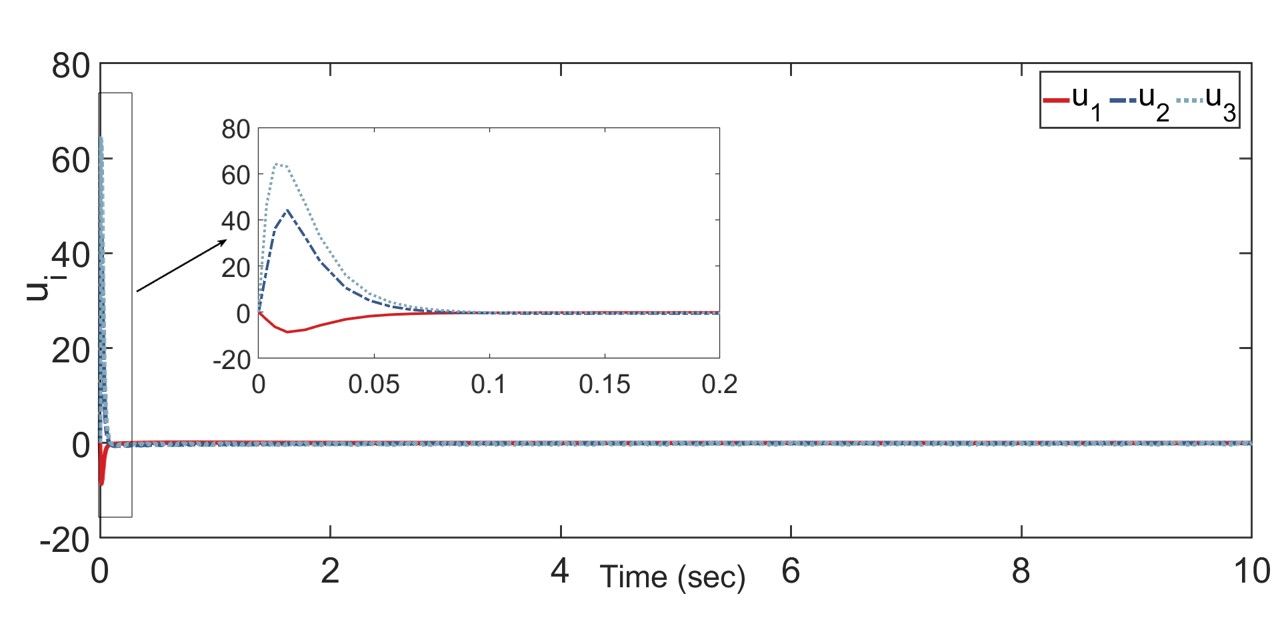}}
	\subfloat[]{\label{huangx}
		\includegraphics[width=0.5\textwidth]{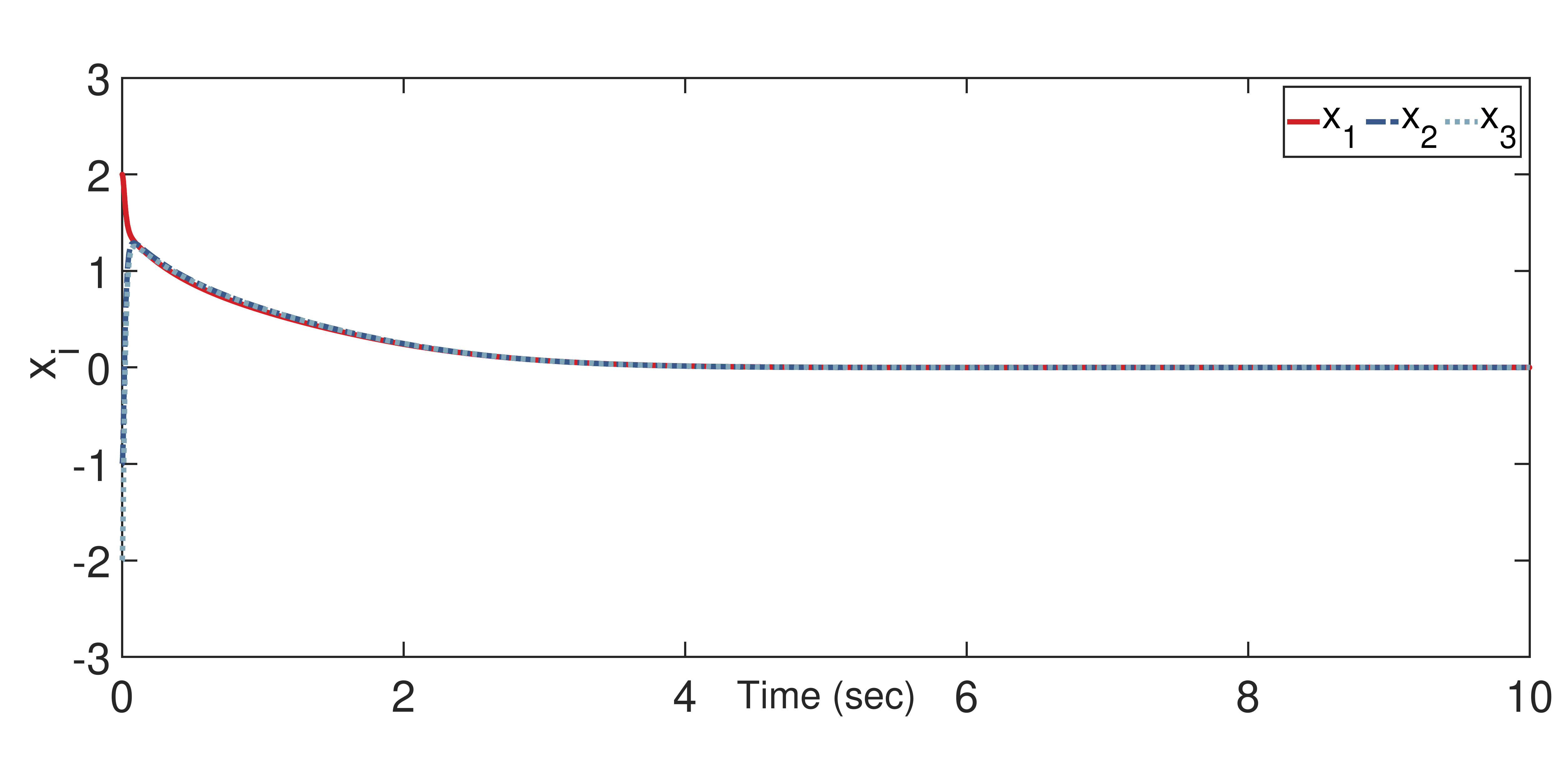}}	
	
	\subfloat[]{\label{huangchi}
		\includegraphics[width=0.5\linewidth]{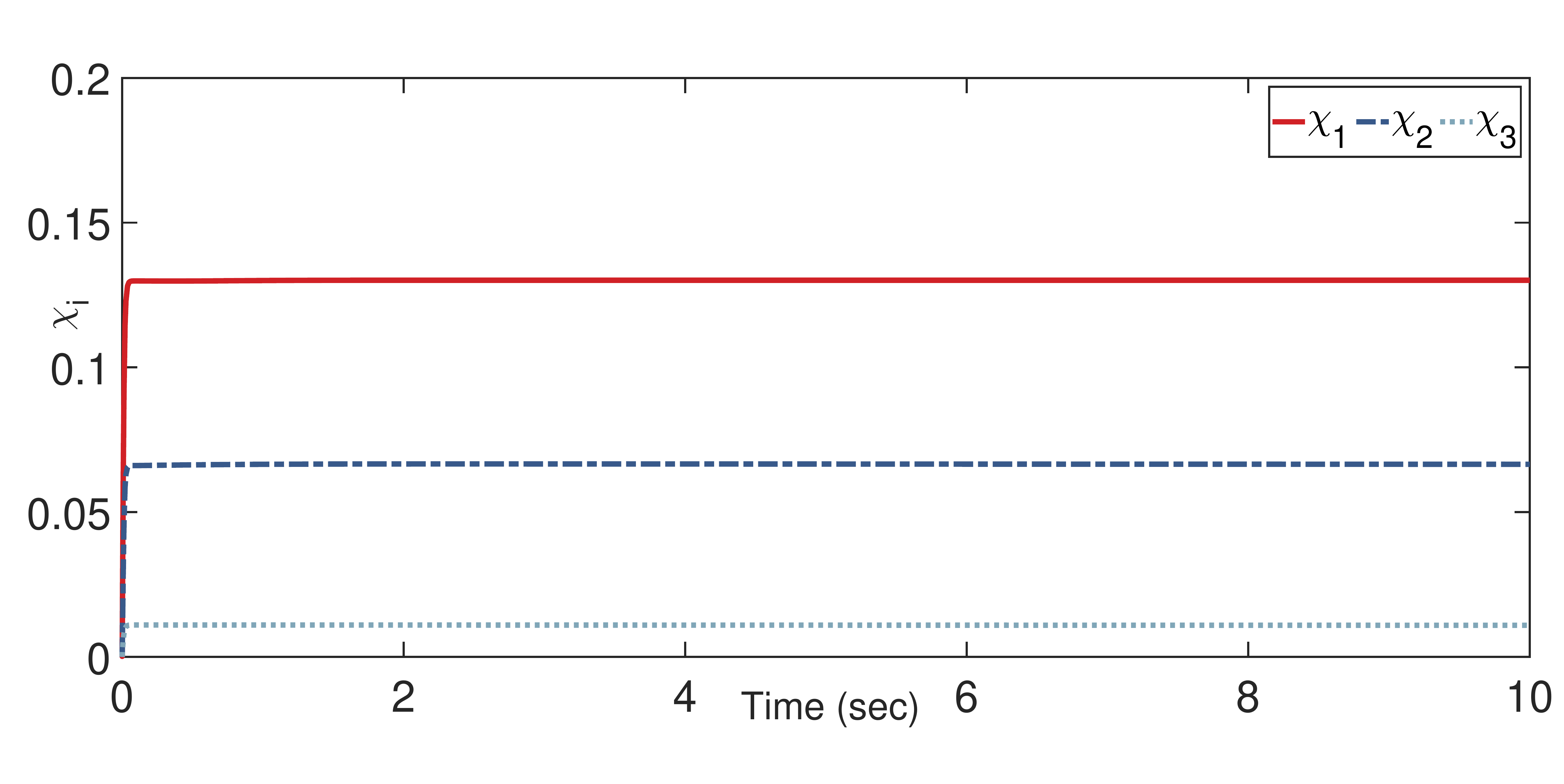}}
	\subfloat[]{\label{huangN}
		\includegraphics[width=0.5\textwidth]{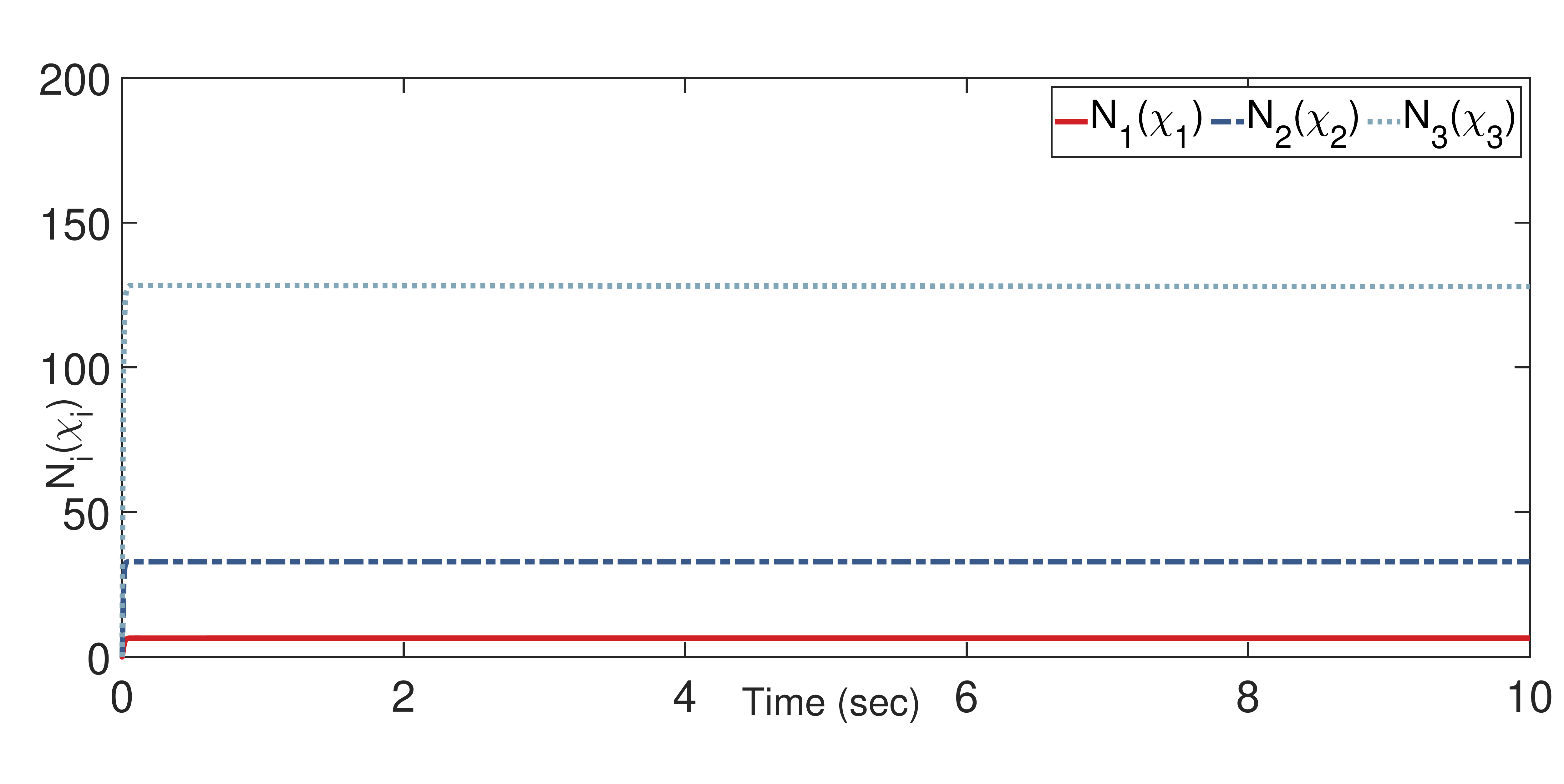}}
	\caption{Results by the novelly proposed Nussbaum based approach.		\protect\subref{huangu} Control input $u_i$, \quad
		\protect\subref{huangx} State $x_i$,
		\protect\subref{huangchi} Updating gain $\chi_i$,\quad \protect\subref{huangN} Traditional Nussbaum function $N_i(\chi_i)$.}\label{standard}
\end{figure*}

Similar to \cite{chenci2016Saturated}, we also quantify the transient performance of control inputs by leveraging the maximum absolute index (MAI) as $|u_i|_{M}=\max_{0\le t\le T} |u_i|$ during the overall control process. The details are presented in Table 1. From the table, we can conclude that $|u_1|_M, |u_2|_M$, and $|u_3|_M$ of the proposed control are $92.5\%, 86.7\%$, and $54.5\%$ lower than that of \cite{HUANG2018auto} respectively, which means that the control shock effect has been greatly suppressed, as well as the improved control smoothness and safety. The boundary of control input can be adjusted by selecting appropriate parameter $a_i$.

\begin{table}[htbp]
	\renewcommand{\tablename} 
	\caption{\centering{TABLE I} \protect \\ \qquad \qquad \quad COMPARISONS RESULTS ON MAI}\\
	
\begin{tabular}{|c|c|c|}
		\hline
		MAI & \begin{tabular}[c]{@{}c@{}}Proposed Saturated \\ Nussbaum based Control\end{tabular} & \begin{tabular}[c]{@{}c@{}}Traditional Nussbaum \\ based Control\cite{HUANG2018auto}\end{tabular} \\ \hline
		$|u_1|_M$	& 4.8 & 64.0 \\ \hline
		$|u_2|_M$	& 5.7 & 43.0 \\ \hline
		$|u_3|_M$	& 5.0 & 11.0 \\ \hline
	\end{tabular}
\end{table}

\begin{remark}
	From the above simulation results, it can be found that the consensus control inputs with the novel proposed Nussbaum functions have a finite saturation boundary. Compared with the consensus control inputs using the traditional Nussbaum functions \cite{HUANG2018auto}, the control shock has been greatly suppressed. This effect was not caused by the selection of specific parameters, but a fundamental change in the form of the Nussbaum function design. The boundary of the Nussbaum gain $N_i(\chi_i)$ can be adjusted by changing parameter $a_i$, which in turn affects the boundary of the control shock. Although the time-elongation saturation Nussbaum functions may cause the system to run in the wrong direction for a bit longer time, it is an affordable price for a better transient performance with limited control shocks.
\end{remark}

\section{Conclusion}\label{sec6}

In this paper, the consensus problem of a group of linearly parameterized first-order multi-agent systems with arbitrary non-identical unknown control directions has been investigated. By proposing a novel saturated Nussbaum-type function, distributed adaptive consensus protocols with an adjustable limited control shock are designed based on local neighbor information. With Barbalat's Lemma, it has been theoretically proved that all signals in the closed-loop system are bounded and the average asymptotically consensus of the agents is achieved. A comparative simulation study between novel Nussbaum functions-based approach and traditional  Nussbaum functions \cite{HUANG2018auto} based approach has also been provided to illustrate the effectiveness and the better transient control performance of the novel Nussbaum function based consensus strategies. It is noted that the arbitrary non-identical unknown control directions and control shocks are both handled in this paper, but the more relaxed parameter selection and the robustness of the Nussbaum function based control scheme remain to be investigated for further research.

\begin{ack}                               
This work is supported in part by the National Natural Science Foundation of China under Grants 61503016 and 61977004, in part by the National Key R\&D Program of China (No.2017YFB0103202), in part by the Fundamental Research Funds for the Central Universities (No.YWF-20-BJJ-634, No.YWF-19-BJ-J-259). 

Thanks for the discussions and advice of Dr. Jiangshuai Huang at Chongqing University in China and Dr. Haris. E. Psillakis at National Technical University of Athens in Greece during this work.  
\end{ack}

\bibliographystyle{unsrt}        
\bibliography{autosam}           



\appendix
\section{Proof of Lemma 3.3}    

\begin{proof}
	First, consider the constructed special relationship of $G_i(\chi_i)$ and $G_{i+1}(\chi_{i+1})$ as shown in Fig. 2. Without loss of generality, assume that $M=4$, $T_i=1$, and $sign({\varrho}_i)=1$. Clearly, over the interval $[\chi_{2k-2,i}, \chi_{2k,i}), k=1,2,\ldots$, we have $G_i(\chi_i)<0$ when $\chi_i\in[\chi_{2k-1,i}, \chi_{2k,i})$. Due to $M=4$, we have $n_{k,i+1}=8k$ when $n_{k,i}=2k$, i.e., $\chi_{2k,i}=\chi_{8k,i+1}$.
	
	\textbf{Case 1}: The UCDs of agent $i$ and agent $i+1$ are identical, i.e.,  $sign({\varrho}_{i+1})=sign({\varrho}_{i})$. 
	
	First, as $sign({\varrho}_{i+1})=sign({\varrho}_{i})=1$, the integration value $G_{i+1}(\chi_{i+1})$ is negative on the intervals  $\chi_{i+1}\in[\chi_{8k-3,i+1}, \chi_{8k-2,i+1})$ and $\chi_{i+1}\in[\chi_{8k-1,i+1}, \chi_{8k,i+1})$. Considering that it is not convenient to discuss the range of the right end point $\chi_{n_{k,i+1}}^{out}$ of the interval $[\chi_{n_{k,i+1}}^{in}, \chi_{n_{k,i+1}}^{out}]$, we only investigate the situation that $\chi_i\in[\chi_{8k-3,i+1}, \chi_{8k-2,i+1})$. 
	
	Take $\sin[{\frac{1}{{b_{i+1}}^{8k-3}T_{i+1}}({\chi_{i+1}}-\chi_{8k-3,i+1}})]=0.5$ as an example. After simple calculation, the intervals $[\chi_{n_{k,i+1}}^{in}, \chi_{n_{k,i+1}}^{out}],  k=1,\ldots,\infty$ satisfy the following conditions as 	
	\begin{align}
		\chi_{{8k-3,i+1}}^{in}=\chi_{8k-3,i+1}+\frac{1}{6}{b_{i+1}}^{8k-3}T_{i+1}\pi,\\
		\chi_{{8k-3,i+1}}^{out}=\chi_{8k-2,i+1}-\frac{1}{6}{b_{i+1}}^{8k-3}T_{i+1}\pi,
	\end{align}
	which means that the integration values satisfy the inequality conditions that $G_{i+1}(\chi_{i+1})\le-0.5a_{i+1}T_{i+1}k_{i+1}^{8k-3}$,  $k=1,\ldots,\infty$.
	
	Then, we discuss the range of the interval $[\chi_{n_{k,i}}^{in}, \chi_{n_{k,i}}^{out}]$  on the Nussbaum function of agent $i$, which satisfies that   $G_{i}(\chi_{i})\le-0.5a_{i+1}T_{i+1}{b_{i+1}}^{8k-3}$. Similarly, we have the following conditions as	
	\begin{align}
		\chi_{{2k-1,i}}^{in}=\chi_{2k-1,i}+{b_{i}}^{2k-1}T_{i}arcsin{\frac{1}{2b_{i+1}^{2}}},\\
		\chi_{{2k-1,i}}^{out}=\chi_{2k,i}-{b_{i}}^{n_{k,i}-1}T_{i}arcsin{\frac{1}{2b_{i+1}^{2}}}.
	\end{align}
	Since $b_{i+1}>1$, we can easily calculate that $\chi_{{2k-1,i}}^{in}<\chi_{{8k-3,i+1}}^{in} \quad and \quad \chi_{{2k-1,i}}^{out}>\chi_{{8k-3,i+1}}^{out}$, which means that over the interval $[\chi_{{8k-3,i+1}}^{in}, \chi_{{8k-3,i+1}}^{out}]$, the following inequality conditions exist	
	\begin{align}
		G_{i}(\chi_{i})&\le0.5G_{i+1}^{min}(\chi_{8k-3, i+1})<0,
		\\ G_{i+1}(\chi_{i+1})&\le0.5G_{i+1}^{min}(\chi_{8k-3, i+1})<0.
	\end{align}
	\textbf{Case 2}: The UCDs of agent $i$ and agent $i+1$ are opposite, i.e., $sign({\varrho}_{i+1})=-sign({\varrho}_{i})$.

	Similarly, as $sign({\varrho}_{i+1})=-sign({\varrho}_{i})=-1$, the integration value $G_{i+1}(\chi_{i+1})$ is negative on the intervals  $\chi_{i+1}\in[\chi_{8k-4,i+1}, \chi_{8k-3,i+1})$ and $\chi_{i+1}\in[\chi_{8k-2,i+1}, \chi_{8k-1,i+1})$. Considering that it is not convenient to discuss the range of the left end point $\chi_{n_{k,i+1}}^{in}$ of the interval $[\chi_{n_{k,i+1}}^{in}, \chi_{n_{k,i+1}}^{out}]$, we only investigate the situation that $\chi_i\in[\chi_{8k-2,i+1}, \chi_{8k-1,i+1})$.
	
	Take $\sin[{\frac{1}{{b_{i+1}}^{8k-3}T_{i+1}}({\chi_{i+1}}-\chi_{8k-3,i+1}})]=0.5$ as an example. After simple calculation, the intervals $[\chi_{n_{k,i+1}}^{in}, \chi_{n_{k,i+1}}^{out}],  k=1,\ldots,\infty$ satisfy the following conditions as	
	\begin{align}
		\chi_{{8k-2,i+1}}^{in} & =\chi_{8k-2,i+1}+\frac{1}{6}{b_{i+1}}^{8k-2}T_{i+1}\pi,\\
		\chi_{{8k-2,i+1}}^{out} & =\chi_{8k-1,i+1}-\frac{1}{6}{b_{i+1}}^{8k-2}T_{i+1}\pi,
	\end{align}
	which means that the integration values satisfy the inequality conditions that $G_{i+1}(\chi_{i+1})\le-0.5a_{i+1}T_{i+1}k_{i+1}^{8k-2}$,  $k=1,\ldots,\infty$.
	
	Similarly, we have the following conditions for the Nussbaum function of $i$th agent as	
	\begin{align}
		\chi_{{2k-1,i}}^{in}=\chi_{2k-1,i}+{b_{i}}^{2k-1}T_{i}arcsin{\frac{1}{2b_{i+1}}},\\
		\chi_{{2k-1,i}}^{out}=\chi_{2k,i}-{b_{i}}^{n_{k,i}-1}T_{i}arcsin{\frac{1}{2b_{i+1}}}.
	\end{align}
	Since $b_{i+1}>1$, we can easily calculate that $\chi_{{2k-1,i}}^{in}<\chi_{{8k-2,i+1}}^{in} \quad and \quad \chi_{{2k-1,i}}^{out}>\chi_{{8k-2,i+1}}^{out}$, which mean that over the interval $[\chi_{{8k-2,i+1}}^{in}, \chi_{{8k-2,i+1}}^{out}]$, the following inequality conditions exist	
	\begin{align}
		G_{i}(\chi_{i})&\le0.5G_{i+1}^{min}(\chi_{8k-2, i+1})<0,
		\\ G_{i+1}(\chi_{i+1})&\le0.5G_{i+1}^{min}(\chi_{8k-2, i+1})<0.
	\end{align}
	
	In conclusion, for agent $i$ and $i+1$, over the interval $[\chi_{2k-2,i}, \chi_{2k,i})$, there must exist an interval  $[\chi_{n_{k,i+1}}^{in}, \chi_{n_{k,i+1}}^{out}]$ on which $sign({\varrho}_{i+1})G_{i+1}({\chi_i+1}) \le0.5G_{i+1}^{min}(\chi_{i+1})<0$ and  $sign({\varrho}_i)G_i({\chi_i}) \le0.5G_{i+1}^{min}(\chi_{i+1})<0$. The maximum value of the integration value $G_{i}(\chi_{i})$ depends on the minimum value of  the integration value $G_{i+1}(\chi_{i+1})$.
	
	The proof is completed.
\end{proof}
\end{document}